\documentclass[APS]{revtex4-2}

\usepackage{hyperref}
\usepackage{amssymb,amsmath,graphicx,stmaryrd,mathtools, amsthm, color}

\usepackage{algorithm,algpseudocode}
\usepackage{enumitem,subcaption}

\newtheorem{theorem}{Theorem}
\newtheorem{lemma}[theorem]{Lemma}

\theoremstyle{remark}

\newcommand{\ket}[1]{\left|#1\right\rangle}
\newcommand{\bra}[1]{\left\langle#1\right|}

\newcommand{\wsize}{{\rm{wsize}}}

\newcommand{\cT}{\mathcal T}

\newcommand{\cV}{\mathcal V}

\newcommand{\black}{{\rm{black}}}
\newcommand{\red}{{\rm{red}}}

\begin{document}

\title{Simplified Quantum Algorithm for the Oracle Identification Problem}
\author{Leila Taghavi}
\email[E-mail me at: ]{l.taghavi@gmail.com}
\affiliation{QuOne Lab, Phanous Research and Innovation Centre, Tehran, Iran}

\begin{abstract}
In the oracle identification problem  we have oracle access to bits of an unknown string $x$ of length $n$, with the promise that it belongs to a known set $C\subseteq\{0,1\}^n$. The goal is to identify $x$ using as few queries to the oracle as possible. 
We develop a quantum query algorithm for this problem with query complexity
$O\left(\sqrt{\frac{n\log M }{\log(n/\log M)+1}}\right)$, where $M$ is the size of $C$. This bound is already derived by Kothari in 2014, 
for which we provide a more elegant simpler proof. 
\end{abstract}

\maketitle

\section{Introduction}
In the oracle identification problem, we have query access to bits of an unknown $n$-bit string $x$, with the promise that it belongs to a known set $C\subseteq\{0,1\}^n$ of size $M$. We want to determine $x$ while minimizing the number of queries to bits of $x$. We denote this problem by OIP$(C)$.  Here, we are interested in quantum algorithms for this problem in which case we assume we can query bits of $x$ in superposition.
 
We are interested in the query complexity of OIP$(C)$ in the worst case with the promise that $|C|=M$. To this end, for a given $M$ and $n$, 
we say that $Q\big({\rm OIP}(M,n)\big)=q$ if there exists a bounded error quantum query algorithm  that solves any OIP$(C)$ where $|C|=M$
with at most $q$ queries.
We note that for some sets $C$ the quantum query complexity might be less than $Q(\text{OIP}(M,n))$, yet here we consider the worst case query complexity over such choices of $C$.

Characterization of $Q(\text{OIP}(M,n)\big)$ for $M\leq n$ is easy; by reducing the problem to Grover's search on $M$ elements we find that $Q(\text{OIP}(M,n)\big)=O(\sqrt M)$~\cite{CKOR13}. For the hard case of $n<M\leq 2^n$,
Kothari~\cite{Kot14} characterized the quantum query complexity of OIP$(M,n)$ by proposing an algorithm beyond a simple Grover's search.  
This algorithm is based on ideas from classical learning theory in combination with a composition property of the so called filtered $\gamma _2$-norm. Here, we give a simpler direct proof of the same result that eliminates the need for the filtered $\gamma_2$-norm. Our proof is based on the framework of~\cite{BT20} that converts a classical algorithm into an improved quantum one.  

Assume that we have a classical algorithm that computes a function $f:[\ell]^n\to[m]$ with query complexity $T$. Moreover, assume that we have a \emph{guessing algorithm} that tries to predict values of queried bits, making at most $G$ mistakes. Using these classical algorithms, we can design a quantum query algorithm for computing $f$ with query complexity $O(\sqrt{GT})$~\cite{BT20, LL16}. This result in~\cite{BT20} is proven based on the framework of non-binary span programs (NBSP)~\cite{BT19}, which is a generalization of the span programs~\cite{Rei09} for functions with non-binary input and/or output alphabets. In this paper, once again using the framework of NBSPs we prove a generalization of the aforementioned result of~\cite{BT20}. Next, using that generalization we give a simple proof of the following bound on $Q(\text{OIP}(M,n)\big)$.

 \begin{theorem} \label{thm:OIP}
Suppose that we are given a set $C\subseteq \{0,1\}^n$, where $|C|=M$ and $n<M\leq 2^n$. Also suppose that we have query access to bits of a string $x\in \{0,1\}^n$ with the promise that $x\in C$. Then quantum query complexity of identifying $x$ is $O\left(\sqrt{\frac{n\log M }{\log(n/\log M)+1}}\right)$.
 \end{theorem}

Note that the bound of this theorem is tight; it is shown by Kothari~\cite{Kot14} that there exists a set $C$ of size $M$ for which $Q\big(\text{OIP}(C)\big)=\Omega\left(\sqrt{\frac{n\log M }{\log(n/\log M)+1}}\right)$.

\section{Classical to quantum query algorithm}\label{sec:classical2quantum}
A classical query algorithm can be modeled using a \emph{decision tree}. A decision tree is a directed acyclic graph which depicts the sequence of queried bits along the algorithm. Any internal node of this tree is labeled by an index $i\in[n]$, and output edges of this node are labeled using possible outcomes of the query of $x_i$. 
Furthermore, every leaf of the decision tree is labeled by an output value of the algorithm. Let $P_x$ be the unique path from the root to a leaf, for which the labels of edges matches the bits of $x$. 

We consider \emph{guessing algorithms} along with decision trees.
A guessing algorithm is a classical algorithm that predicts the values of queries. Such an algorithm can be represented by an edge-coloring of the decision tree. In this coloring we color any edge associated to an output of the guessing algorithm black and color the rest of edges red. We call such a coloring of edges of a decision tree a \emph{G-coloring}. More formally, a G-coloring is a coloring of the edges of the decision tree using two colors black and red in such a way that every node has exactly one outgoing edge with black color. See Figure~\ref{fig:twoZeroesTree} for an example of a decision tree and its G-coloring.

\begin{figure}
  \includegraphics[scale=1,bb=190 0 165 165]{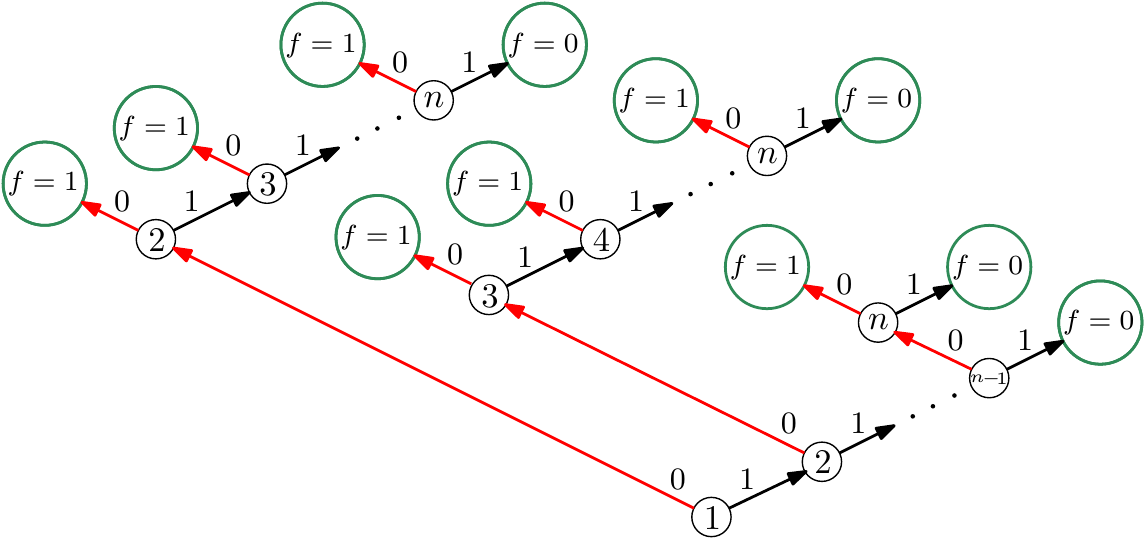}  
  \caption{\label{fig:twoZeroesTree} A decision tree for a function that outputs $1$ if the input string contains at least two $0$s. Any leaf of the tree is labeled with its associated output, and other vertices are labeled with the index that the classical algorithm queries in that vertex. A G-coloring is represented by the edge coloring of the decision tree using two colors black and red. This guessing algorithm always guesses that the output of the queried bit is $1$ (black edges). Edges with label $0$ are colored red and indicate wrong guesses of the guessing algorithm. The depth of this tree is $T=n$ which is the classical query complexity of this algorithm for this function. The maximum number of wrong guesses of the guessing algorithm is $G=2$ since after seeing the first two mistakes (red edges) the output of the function will be revealed.} 
\end{figure}

To prove our main result, here we first give a bound on the quantum query complexity of a problem based on a decision tree and a G-coloring on it. This result is proven based on ideas from~\cite{BT20}.

\begin{theorem}\label{thm:binaryClassical2quantumW}
Assume that we have a decision tree $\mathcal T$ for a function $f:D_f\to [m]$ with $D_f\subseteq \{0,1\}^n$ whose depth is $T$. Furthermore, assume that for a G-coloring of the edges of $\mathcal T$, the number of red edges in each path from the root to leaves of $\mathcal T$ is at most $G$.
Let $G_x$ be the number of red edges in $P_x$, and for $1\leq g\leq G_x$, let $T_{g,x}$ be the number of black edges in $P_x$ after the $g$-th red edge and before the next red one. Also let $T_{0, x}$ be the number of black edges before the first red edge in $P_x$.
 Then there exists a bounded error quantum query algorithm that computes the function $f$ with query complexity 
 \begin{equation}
     O\left(\max_x \sum_{g=0}^{G_x}\sqrt{T_{g,x}}\right).
 \end{equation}
\end{theorem}
Before getting into the proof of this theorem, let us make an intuition about its statement. Think of the guessing algorithm as a reference that answers to our queries. We look for its mistakes and  truncate the classical algorithm into $G+1$ parts, where $G$ is the number of mistakes that the guessing algorithm makes. In each of these truncated pieces, the guessing algorithm predicts all the necessary queries correctly, so in our quantum algorithm we can follow the classical one without making any queries. To count the number of queries, note that there exists a variant of Grover search algorithm that finds the first marked element in a list of $n$ elements making $O\left(\sqrt{j}\right)$ queries, where $j$ is the index of this marked element~\cite{Kot14}. Using this result we can find the $i$th mistake of the guessing algorithm using  $O\left(\sqrt{T_{i-1,x}}\right)$ queries. The total number of queries will then be $\sum_{g=0}^{G_x}\sqrt{T_{g,x}}\log n$, where the extra $\log n$ factor is needed for error reduction\footnote{The Grover search algorithm is a bounded error algorithm meaning that the probability of getting a correct answer  is at least $\frac23$. By successive calls of this algorithm on different sets, the errors in different outputs aggregate  and the probability of getting the correct final answer becomes increasingly small. To reduce this error we can repeat each Grover call $\log n$ times and get the majority vote. This makes the total error bounded with the cost of a $\log n$ factor}. In our proof we design a non-binary span program for this problem that mimics the behavior of the stated algorithm without getting this extra $\log n$ factor in its complexity.

\begin{figure}
 \includegraphics[scale=1,bb=320 0 165 165]{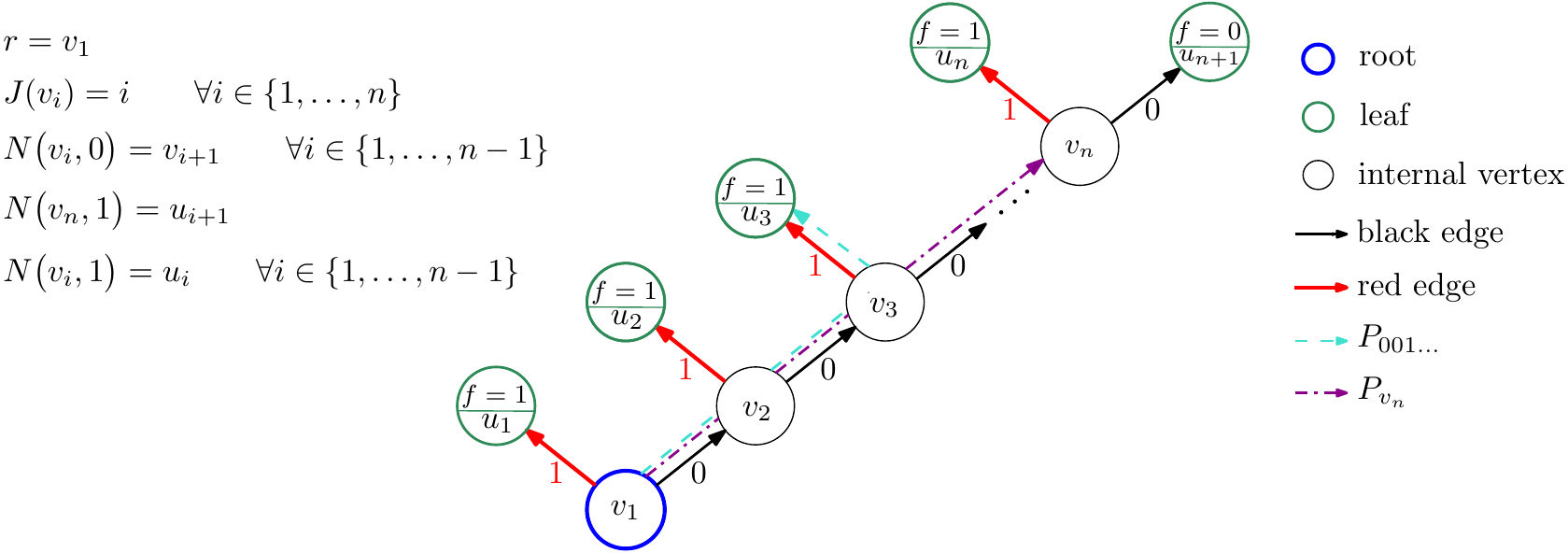}  
\caption{\label{fig:OrTree} A decision tree for the OR function. The output of this function is the OR of the input bits, so the output is 1 if and only if there exists an input index with value 1. The guessing algorithm always guesses that the output of the queried bit is 0 (black edges). Red edges (wrong guesses of the guessing algorithm) have label 1. The depth of this tree is $T=n$, which is the classical query complexity of this algorithm for the OR function. The maximum number of wrong guesses of the guessing algorithm is $G=1$, since after seeing the first mistake (red edge) the output of the function will be revealed. In the left, you can see all notations that are defined in the proof of Theorem~\ref{thm:binaryClassical2quantumW}.} 
\end{figure}

 To present this span program first we need to develop some notations. Let $V(\mathcal T)$ be the vertex set of $\mathcal T$. Then for every internal vertex $v\in V(\cT)$, its associated index is denoted by $J(v)$, i.e., $J(v)$ is the index $1\leq j\leq n$ that is queried by the classical algorithm at node $v$. The two outgoing edges of $v$ are indexed by elements of $\{0,1\}$ and connect $v$ to two other vertices. We denote these vertices by $N(v, 0)$ and $N(v, 1)$. That is, $N(v, q)$, for $q\in \{0,1\}$, is the next vertex that is reached from $v$ after following the outgoing edge with label $q$. We also represent the G-coloring of edges of $\cT$ by a function $C(v, q)\in \{\black, \red\}$ where $v$ is an internal vertex, $q\in \{0,1\}$. $C(v, q)$ is the color of the outgoing edge of $v$ with label $q$. Also let $b(v)$ be the number of black edges after the last red edge in the path from the root to the vertex $v$. These notations are depicted in an example in Figure~\ref{fig:OrTree}.

\begin{proof}
For every $x\in D_f$ there is an associated leaf of the tree $\cT$ that is reached once we follow edges of the tree with labels $x_j$, starting from the root. In order to find $f(x)$ it suffices to find this associated leaf because this is what the classical query algorithm does; once we find the leaf associated to $x$, we find the path that the classical query algorithm would take and then find $f(x)$. Thus in order to compute $f$, we may compute another function $\tilde f$ which given $x$ outputs its associated leaf of $\mathcal T$, and to prove an upper bound  on the quantum query complexity, it suffices to design an algorithm for $\tilde f$. 

We use the framework of non-binary span programs (NBSPs) for bounding the quantum query complexity of $\tilde f$. For more details on NBSPs we refer to~\cite{BT19} and here we only present the ingredients of an NBSP through the example of our particular function $\tilde f$:  
\begin{itemize}
\item The first ingredient of an NBSP is a finite-dimensional vector space  $\cV$ that is called the input space. Here, in our problem the input space is determined by the orthonormal basis indexed by vertices of the decision tree $\mathcal T$:
\begin{equation}
    \{\ket{v}\,|\, v\in V(\cT)\},
\end{equation} 
\item An NBSP contains some target vectors $|t_0 \rangle, |t_2 \rangle,\ldots ,|t_{m-1} \rangle\in \cV$, one for any possible value of the function. Here, the output values of $\tilde f$ are indexed by leaves $u$ of $\cT$ and we let 
\begin{equation}
    \ket{t_u}=\ket{r}-\ket{u},
\end{equation}
where $r\in V(\cT)$ is the root of the tree.

\item The input vectors $I_{j,q}$ of an NBSP are some subsets $I_{j,q}\subseteq \cV$ for every $1\leq j\leq n$ and $q\in \{0,1\}$. Here, the input sets are defined by
\begin{equation}
    I_{j,q}=\left\{\sqrt{W_{C(v, q),b(v)}}\big(\ket{v}-\ket{{N(v, q)}}\big)  \,\Big|\, \forall v\in V(\cT) \text{ s.t. } J(v)=j \right\},
\end{equation}
where $W_{\black,b}$ and $W_{\red,b}$ are positive real numbers to be determined~\footnote{Note that the weight of every edge $(v,q)$  not only depends on its color but also depends on $b(v)$, that is the number of black edges after the last red edge in the path from root to $v$.}.
\end{itemize}

Having all these ingredients, we
let $I\subseteq \cV$ be 
\begin{equation}
    I=\bigcup_{j=1}^n \bigcup_{q\in [\ell]} I_{j,q},
\end{equation}
and for every $x\in D_f$ we define the set of \emph{available vectors} $I(x)$  by
\begin{equation}
  I(x)=\bigcup_{j=1}^nI_{j,x_j}.  
\end{equation}
Now we say that the above NBSP evaluates $\tilde f$ if for every $x$ the target vector $\ket{t_{\alpha}}$ belongs to the span of available vectors $I(x)$ if and only if $\alpha=\tilde f(x)$. Furthermore, there should be \emph{negative} and \emph{positive witnesses} for this, that are explained below.

For every vertex $v$ of $\cT$, let $P_v$  be the (unique) path from the root $r$ to the vertex $v$. 
Then, for  every $x\in D_f$ there exists a path $P_x=P_{\tilde f(x)}$ from the root of the decision tree to the leaf $\tilde f(x)$. 

Thus, the target vector $\big|t_{\tilde f(x)}\big\rangle$ equals
\begin{align}\label{eq:positive-w-vector}
\big|t_{\tilde f(x)}\big\rangle= \ket r - \big|\tilde f(x)\big\rangle=\sum_{v\in P_{x}} \frac{1}{\sqrt{W_{C\left(v, x_{J(v)}\right),b(v)}}} \left\{ \sqrt{W_{C\left(v, x_{J(v)}\right),b(v)}} \left(\ket{v}-\ket{N(v,x_{J(v)})}\right)\right\},
\end{align}
where the vectors in the braces are all available for $x$. This shows that $\ket{t_{\tilde f(x)}}$ is the only target vector belonging to the span of $I(x)$ 
We show that there exists a negative witness for it, defined by 
\begin{equation}
    \ket{\bar{w}_x}=\sum_{v \in P_{x} } \ket{v}\in \cV.
\end{equation}
We note that that $\ket{\bar{w}_x}$ is orthogonal to all available vectors in $I(x)$. Moreover, it satisfies
$\bra{\bar w_x} t_\alpha\rangle=\bra{\bar w_x} r\rangle=1$ for every $\alpha\neq \tilde{f}(x)$. These two facts ensure that $\ket{t_\alpha}$ for $\alpha\neq \tilde f(x)$ does not belong to the span of $I(x)$.

Now based on the results of~\cite{BT19} the above NBSP gives a bound on the query complexity of $\tilde f$. To evaluate this bound we need to estimate two quantities called the \emph{positive complexity} and \emph{negative complexity} of the NBSP. The positive complexity denoted by $\mathrm{wsize}^+(P,w,\bar{w})$ is the maximum of the squared norm of the coefficient-vector in the expansion~\eqref{eq:positive-w-vector}.
For any $x$ this norm is computed as
\begin{equation}
\wsize^+(x)=\sum_{g=0}^{G_x}\left(\frac{1}{W_{\red,T_{g,x}}} +\sum_{b=0}^{T_{g,x}-1}\frac{1}{W_{\black,b}}\right).
\end{equation}

The negative complexity for any 
$x$ denoted by $\wsize^-(x)$  is equal to the sum of the squared inner product of $\ket{\bar w_x}$ with all vectors in $I$. To compute this quantity we need to compute the overlap of $\ket{\bar w_x}$ with all vectors of the form
\begin{equation}
    \sqrt{W_{C(v, q),b(v)}}\Big(\ket{v}-\ket{{N(v, q)}}\Big).
\end{equation}
We note that such a vector
contributes in the negative complexity (is not orthogonal to $\ket{\bar w_x}$) only if its corresponding edge $\{v, N(v, q)\}$ leaves the path $P_x$, i.e., they have only the vertex $v$ in common. In this case, the contribution would be equal to $W_{C(v, q),b(v)}$, which is the weight of that edge. Therefore, we have 
\begin{equation}
\wsize^-(x)=\sum_{g=0}^{G_x}\left(W_{\black,T_{g,x}} +\sum_{b=0}^{T_{g,x}-1}W_{\red,b}\right).
\end{equation}

Now letting $W_{\red,b}=\frac{1}{W_{\black,b}}=\sqrt{b+1}-\sqrt{b}$, both the positive and negative witnesses are bounded by
\begin{equation}
\sum_{g=0}^{G_x}\left(\frac{1}{\sqrt{T_{g,x}+1}-\sqrt{T_{g,x}}}+\sum_{b=0}^{T_{g,x}-1}\sqrt{b+1}-\sqrt{b}\right)
= O\left( \sum_{g=0}^{G_x}\sqrt{T_{g,x}}\right),
\end{equation}
where we used the fact that 
\begin{align}
\left(\sqrt{T_{g,x}+1}-\sqrt{T_{g,x}}\right)^{-1} =&
\left(\sqrt{T_{g,x}}\Big(\sqrt{\frac{T_{g,x}+1}{T_{g,x}}}-1\Big)\right)^{-1} \\=& \;
O\left(\frac{1}{\sqrt{T_{g,x}}}\Big(1+\frac{1}{2T_{g,x}}-1\Big)^{-1}\right)=
O\left(\sqrt{T_{g,x}}\right).
\end{align}
Putting these together the complexity of the NBSP is defined and bounded by
\begin{align}\label{eq:wsize-pwwb}
 \mathrm{wsize}=\sqrt{\max_x \mathrm{wsize}^+(x)\,\cdot\,\max_x\mathrm{wsize}^-(x)}=O\left(\max_x \sum_{g=0}^{G_x}\sqrt{T_{g,x}}\right). 
\end{align}
This is also a bound on the quantum query complexity of $\tilde f$, and that of $f$. 
\end{proof}

\section{Proof of the Main Result}
Recall that in the oracle identification problem  we are given a set $C\subseteq \{0,1\}^n$, where $|C|=M$, and query access to a string $x\in \{0,1\}^n$, with the promise that $x\in C$. We want to identify the input $x$ using the least number of queries to the input oracle. We use Theorem~\ref{thm:binaryClassical2quantumW} to design a quantum query algorithm for this problem. To apply this theorem we need a classical algorithm and a guessing algorithm, for which we borrow ideas from~\cite{Kot14}.  The classical algorithm chooses the order in which we query bits of $x$, and the guessing algorithm predicts the queries outputs based on the set $C$. In the beginning, any element in $C$ is a possible candidate for $x$. We need to choose an index of $x$ and predict a value for it in such a way that an incorrect guess eliminates the biggest portion of members of $C$ from our possible candidates. We continue removing impossible candidates form $C$ until we are left with a single element which is equal to $x$. Using this process we can find $x$ with minimum number of incorrect guesses.

The following lemma indicates this special order of queries and their predicted outcomes.

\begin{lemma}[\cite{Heg95}]\label{lem:oip-set} For any set $C\in\{0,1\}^n$, there exists a string $s\in\{0,1\}^n$ and a permutation $\pi:\{1,\ldots,n\}\to\{1,\ldots,n\}$, such that for any $j\in\{1,\ldots,n\}$, we have $|C_j|\leq \frac{|C|}{\max\{2,j\}}$, where 
\begin{equation}
    C_j=\big\{c\in C\big|\; c_{\pi(j)}\neq s_{\pi(j)}, c_{\pi(i)}=s_{\pi(i)} \; \forall i<j\big\}
\end{equation}
\end{lemma}

Now we have all ingredients to prove Theorem~\ref{thm:OIP}.

\begin{proof}[Proof of Theorem~\ref{thm:OIP}]
The classical algorithm for this problem is as follows.
 We first apply Lemma~\ref{lem:oip-set} to the set $C^{(1)}=C$ and obtain a string $s^{(1)}$ and a permutation $\pi_1$. We then query the indices one by one using the order defined by the permutation $\pi_1$ (the $i$-th query is the index $\pi_1(i)$). The guessing algorithm predicts the value of this query to be  $s^{(1)}_{\pi_1(i)}$. We continue to query until we see a mismatch between a queried value and the corresponding bit of $s^{(1)}$, that corresponds to a wrong guess. Suppose that the first mismatch is in the $p_1$-th query, so we have $x_{\pi_1(p_1)}\neq s^{(1)}_{\pi_1(p_1)}$. It is now guaranteed that $x\in C^{(1)}_{p_1}$ and we can eliminate at least a fraction $1-\frac{1}{\max\{2,p_1\}}$ of members of $C^{(1)}$ from possible candidates for $x$. We then apply Lemma~\ref{lem:oip-set} once again to the set $C^{(2)}=C^{(1)}_{p_1}$ and obtain a new string $s^{(2)}$ and a new permutation $\pi_2$. We then continue to query using the order defined by $\pi_2$ until we find a mismatch between a queried value and the string $s^{(2)}$. This again eliminates some members of $C^{(2)}$.  We continue this process until we reach some set $C^{(r)}$ with $|C^{(r)}|=1$ and output the only member of $C^{(r)}$ as the value of $x$. 

In the above algorithm, the number of wrong guesses (red edges observed in the decision tree) is $r$.
Moreover, using Theorem~\ref{thm:binaryClassical2quantumW}, the quantum query complexity of this problem is $\sum_{i=1}^{r-1}\sqrt{p_i}$. Indeed, to find an estimate on the quantum query complexity, we need to compute the maximum of this quantity over all possible choices of $r$ and $p_i$'s. To this end,
note that the total number of black edges in any path from the root to a leaf is at most $n$ and 
\begin{equation}\label{eq:con1}
\sum_{i=1}^{r-1} p_i \leq n.
\end{equation}
We also using Lemma~\ref{lem:oip-set} know that the $i$-th wrong guess reduces the size of $C^{(i)}$ by a factor of $\max\{2,p_i\}$. Therefore, we have \begin{align}\label{eq:con2}
\prod_{i=1}^{r-1} \max\{2, p_i\}\leq M.
\end{align}
Thus, using Theorem~\ref{thm:binaryClassical2quantumW} the quantum query complexity of this problem is the maximum value of \begin{equation}
    \sum_{i=1}^{r-1}\sqrt{p_i},
\end{equation}
subject to the constraints~\eqref{eq:con1} and~\eqref{eq:con2}.
Kothari showed that the optimal value of this optimization problem is bounded by $O\left(\sqrt{\frac{n\log M}{\log(n/\log M)+1}}\right)$. The proof is by finding the optimal value of  another linear program that upperbounds the optimal value of this optimization problem and is easier to solve (see  Appendix B of the full version of~\cite{Kot14} for more details). 

Putting these together we conclude that the quantum query complexity of the oracle identification problem is $O\left(\sqrt{\frac{n\log M}{\log(n/\log M)+1}}\right)$.
\end{proof}

As an example assume that $C=\{0000,0001,0011,0111,1111\}$. Figure~\ref{fig:Cexample} shows how using Theorem~\ref{thm:OIP} we can discover any given input $x\in C$. In the left we see possible updates in the set $C$ along the algorithm based on queries. In the right the associated decision tree is depicted. Note that similar colors in the left picture and the decision tree determines the same stage of the algorithm, for example in the root of the decision tree which we show using color blue, the set $C$ is the leftmost set in the left picture which has blue color too. We also color the index that we query in each step of the algorithm by red color in $C$. Vertices with green color are the leaves of the decision tree, which we label by an input $x$.

\begin{figure}
  \includegraphics[scale=.9,bb=200 0 340 340]{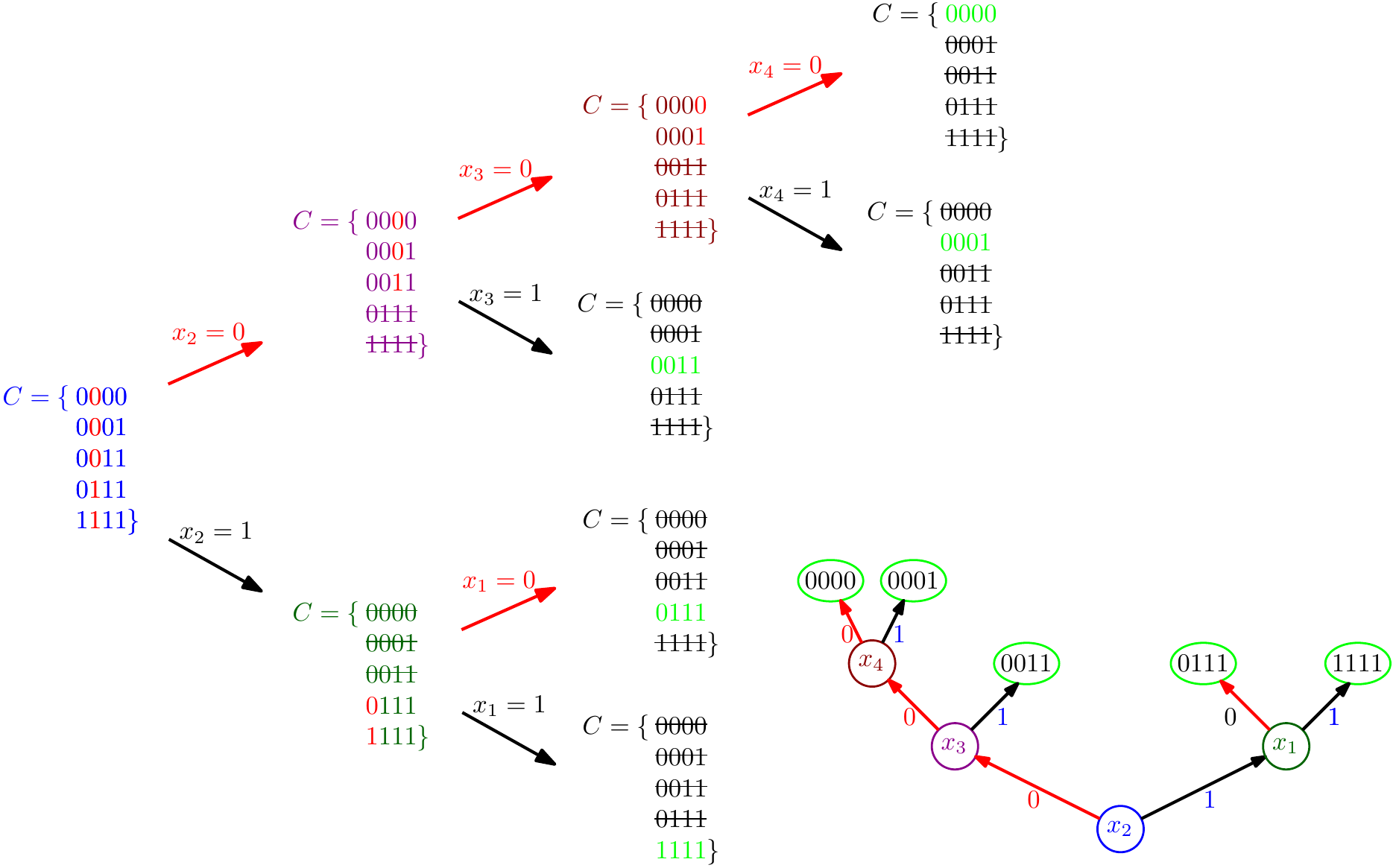}  
\caption{\label{fig:Cexample} An example for the decision tree and its coloring based on the algorithm of Theorem~\ref{thm:OIP}. In this example we assume that the input set $C$ is equal to $\{0000,0001,0011,0111,1111\}$. In the left we see the updates of the set $C$ after each query, the order in which we query is determined using Lemma~\ref{lem:oip-set}. For any update of the set $C$, the queried index which is the most informative bit based on lemma~\ref{lem:oip-set} is colored red. In the right the associated decision tree is depicted. Note that any colored node of the decision tree is associated to the set $C$ in the left with the same color.}  
\end{figure}

To conclude, note that the classical algorithm that we used here is based on Khotari's~\cite{Kot14}. This algorithm is based on repeated applications of the Grover search algorithm to find  the mismatches between $x$ and $s^{(i)}$ and then uses properties of the $\gamma_2$-norm to eliminate the $\log$ factor that has been added for error reduction. Using our approach this $\log$ factor does not appear in the query complexity in the first place. This can be generalized to any problem of this form, i.e., for any algorithm with repeated calls to Grover search we can use the same idea to eliminate the $\log$ factor we get for error reduction.

\paragraph{Acknowledgements} I would like to thank Salman Beigi for his helpful discussion and suggestions.

\bibliography{references}

\begin{thebibliography}{1}

\bibitem{BT19}
Salman Beigi and Leila Taghavi.
\newblock {Span Program for Non-binary Functions}.
\newblock {\em Quantum Information and Computation(QIC)},
  19(9{\&}10):0760--0792, 2019.

\bibitem{BT20}
Salman Beigi and Leila Taghavi.
\newblock Quantum {S}peedup {B}ased on {C}lassical {D}ecision {T}rees.
\newblock {\em {Quantum}}, 4:241, March 2020.

\bibitem{CKOR13}
Andrew~M. Childs, Robin Kothari, Maris Ozols, and Martin Roetteler.
\newblock Easy and hard functions for the boolean hidden shift problem.
\newblock In {\em Leibniz International Proceedings in Informatics, LIPIcs},
  volume~22, pages 50--79, 2013.

\bibitem{Heg95}
Tibor Heged{\H{u}}s.
\newblock Generalized teaching dimensions and the query complexity of learning.
\newblock In {\em Proceedings of the eighth annual conference on Computational
  learning theory}, pages 108--117, 1995.

\bibitem{Kot14}
Robin Kothari.
\newblock {An optimal quantum algorithm for the oracle identification problem}.
\newblock In {\em Leibniz International Proceedings in Informatics, LIPIcs},
  volume~25, pages 482--493, 2014.

\bibitem{LL16}
Cedric Yen-Yu Lin and Han-Hsuan Lin.
\newblock {Upper bounds on quantum query complexity inspired by the
  Elitzur-Vaidman bomb tester}.
\newblock {\em Theory of Computing}, 12:1--35, 2016.

\bibitem{Note1}
The Grover search algorithm is a bounded error algorithm meaning that the
  probability of getting a correct answer is at least $\protect \frac 23$. By
  successive calls of this algorithm on different sets, the errors in different
  outputs aggregate and the probability of getting the correct final answer
  becomes increasingly small. To reduce this error we can repeat each Grover
  call $\protect \qopname \relax o{log}n$ times and get the majority vote. This
  makes the total error bounded with the cost of a $\protect \qopname \relax
  o{log}n$ factor.

\bibitem{Note2}
Note that the weight of every edge $(v,q)$ not only depends on its color but
  also depends on $b(v)$, that is the number of black edges after the last red
  edge in the path from root to $v$.

\bibitem{Rei09}
Ben~W Reichardt.
\newblock Span programs and quantum query complexity: The general adversary
  bound is nearly tight for every boolean function.
\newblock In {\em Foundations of Computer Science, 2009. FOCS'09. 50th Annual
  IEEE Symposium on}, pages 544--551. IEEE, 2009.

\end{thebibliography}
\bibliographystyle{plain}



\end{document}